\documentclass[pra,twocolumn,showpacs,superscriptaddress]{revtex4-1}

\usepackage{graphicx}
\usepackage{color}

\usepackage{amsthm,amsmath,amssymb}
\usepackage{bbold}
\usepackage[matrix,frame,arrow]{xy}
\usepackage{amsmath}
\newcommand{\bra}[1]{\left\langle{#1}\right\vert}
\newcommand{\ket}[1]{\left\vert{#1}\right\rangle}
\newcommand{\qw}[1][-1]{\ar @{-} [0,#1]}

\newcommand{\gate}[1]{*{\xy *+<.6em>{#1};p\save+LU;+RU **\dir{-}\restore\save+RU;+RD **\dir{-}\restore\save+RD;+LD **\dir{-}\restore\POS+LD;+LU **\dir{-}\endxy} \qw}

\newcommand{\measureD}[1]{*{\xy*+=+<.5em>{\vphantom{\rule{0em}{.1em}#1}}*\cir{r_l};p\save*!R{#1} \restore\save+UC;+UC-<.5em,0em>*!R{\hphantom{#1}}+L **\dir{-} \restore\save+DC;+DC-<.5em,0em>*!R{\hphantom{#1}}+L **\dir{-} \restore\POS+UC-<.5em,0em>*!R{\hphantom{#1}}+L;+DC-<.5em,0em>*!R{\hphantom{#1}}+L **\dir{-} \endxy} \qw}

\newcommand{\multimeasureD}[2]{*+<1em,.9em>{\hphantom{#2}}\save[0,0].[#1,0];p\save !C *{#2},p+LU+<0em,0em>;+RU+<-.8em,0em> **\dir{-}\restore\save +LD;+LU **\dir{-}\restore\save +LD;+RD-<.8em,0em> **\dir{-} \restore\save +RD+<0em,.8em>;+RU-<0em,.8em> **\dir{-} \restore \POS !UR*!UR{\cir<.9em>{r_d}};!DR*!DR{\cir<.9em>{d_l}}\restore \qw}

\newcommand{\multigate}[2]{*+<1em,.9em>{\hphantom{#2}} \qw \POS[0,0].[#1,0];p !C *{#2},p \save+LU;+RU **\dir{-}\restore\save+RU;+RD **\dir{-}\restore\save+RD;+LD **\dir{-}\restore\save+LD;+LU **\dir{-}\restore}
\newcommand{\ghost}[1]{*+<1em,.9em>{\hphantom{#1}} \qw}

\newcommand{\gategroup}[6]{\POS"#1,#2"."#3,#2"."#1,#4"."#3,#4"!C*+<#5>\frm{#6}}

\newcommand{\ustick}[1]{*!D!<0em,-.5em>=<0em>{#1}}

\newcommand{\Qcircuit}[1][0em]{\xymatrix @*=<#1>}

\newcommand{\pureghost}[1]{*+<1em,.9em>{\hphantom{#1}}}
\newcommand{\multiprepareC}[2]{*+<1em,.9em>{\hphantom{#2}}\save[0,0].[#1,0];p\save !C
  *{#2},p+RU+<0em,0em>;+LU+<+.8em,0em> **\dir{-}\restore\save +RD;+RU **\dir{-}\restore\save
  +RD;+LD+<.8em,0em> **\dir{-} \restore\save +LD+<0em,.8em>;+LU-<0em,.8em> **\dir{-} \restore \POS
  !UL*!UL{\cir<.9em>{u_r}};!DL*!DL{\cir<.9em>{l_u}}\restore}
\newcommand{\prepareC}[1]{*{\xy*+=+<.5em>{\vphantom{#1\rule{0em}{.1em}}}*\cir{l^r};p\save*!L{#1} \restore\save+UC;+UC+<.5em,0em>*!L{\hphantom{#1}}+R **\dir{-} \restore\save+DC;+DC+<.5em,0em>*!L{\hphantom{#1}}+R **\dir{-} \restore\POS+UC+<.5em,0em>*!L{\hphantom{#1}}+R;+DC+<.5em,0em>*!L{\hphantom{#1}}+R **\dir{-} \endxy}}

\newcommand{\ketbra}[2]{\ket{#1}\bra{#1}}

\newcommand{\Tr}{\operatorname{Tr}}

\newcommand{\floor}[1]{\lfloor#1\rfloor}
\newcommand{\ceil}[1]{\lceil#1\rceil}
 \newcommand{\hilb}[1]{\mathcal{#1}}

\newtheorem{thm}{Proposition}

\begin{document}

\title{Experimental implementation of unambiguous quantum reading}

\author{Michele \surname{Dall'Arno}}

\affiliation{ICFO-Institut de Ciencies Fotoniques, Mediterranean
  Technology Park, E-08860 Castelldefels (Barcelona), Spain}
\affiliation{Quit group, Dipartimento di Fisica ``A. Volta'', via Bassi 6,
  I-27100 Pavia, Italy}

\author{Alessandro \surname{Bisio}}
\author{Giacomo Mauro \surname{D'Ariano}}

\affiliation{Quit group, Dipartimento di Fisica ``A. Volta'', via Bassi 6,
  I-27100 Pavia, Italy}
\affiliation{Istituto Nazionale di Fisica Nucleare,
  Gruppo IV, via Bassi 6, I-27100 Pavia, Italy}

\author{Martina Mikov\'{a}}
\author{Miroslav Je\v{z}ek}
\author{Miloslav Du{\v{s}}ek}

\affiliation{Department of Optics, Faculty of Science, Palacky University,
  17.~listopadu 12, CZ-77146 Olomouc, Czech Republic}

\date{\today}

\begin{abstract}
  We provide the optimal strategy for unambiguous quantum reading of optical
  memories, namely when perfect retrieving of information is achieved
  probabilistically, for the case where noise and loss are negligible. We
  describe the experimental quantum optical implementations, and provide
  experimental results for the single photon case.
\end{abstract}

\maketitle

\section{Introduction}

In the engineering of optical memories (such as CDs or DVDs) and readers, a
tradeoff among several parameters must be taken into account. High precision
in the retrieving of information is surely an indefeasible assumption, but
also energy requirements, size and weight can play a very relevant role for
applications. Clearly, size and weight of the device increase with the energy,
and using a low energetic radiation to read information reduces the heating of
the physical bit, thus allowing for smaller implementation of the bit
itself. Moreover, many physical media (e.g., superconducting devices)
dramatically change their optical property if the energy flow overcomes a
critical threshold.

In the problem of quantum reading \cite{Pir11, BDD11, Nai11, PLGMB11, Hir11}
of optical devices one's task is to exploit the quantum properties of light in
order to retrieve some classical digital information stored in the optical
properties of a given media, making use of as few energy as possible. The
quantum reading of optical memories was first introduced in \cite{Pir11}. A
realistic model of digital memory was considered, where each cell is composed
of a beamsplitter with two possible reflectivities. A single optical input is
available to the reading device, while the other one introduces thermal noise
in the reading process, so that the problem considered is the discrimination
of two lossy and thermal Gaussian channels. It was shown that, for fixed mean
number of photons irradiated over each memory cell, even in the presence of
noise and loss, a quantum source of light can retrieve more information than
any classical source - in particular in the regime of few photons and high
reflectivities. This provided the first evidence that the use of quantum light
can provide great improvements in applications in the technology of digital
memories such as CDs or DVDs.

In practical implementations noise can sometimes be noticeably reduced
\cite{note:illumination}. On the other hand, in general loss inherently
affects quantum optical setups. Nevertheless, a theoretical analysis of the
ideal, i.e. lossless and noiseless, quantum reading provides a theoretical
insight of the problem and a meaningful benchmark for any experimental
realization.  In this hypothesis quantum reading of optical devices can be
recasted to a discrimination among optical devices with low energy and high
precision.

In the ideal reading of a classical bit of information from an optical memory,
namely in the discrimination of a quantum optical device from a set of two,
different scenarios can be distinguished. A possibility is the on-the-fly
retrieving of information (e.g. multimedia streaming), where the requirement
is that the reading operation is performed fast - namely, only once, but a
modest amount of errors in the retrieved information is tolerable. This
scenario corresponds to the problem of minimum energy ambiguous discrimination
of optical devices \cite{Aci01, DLP02, DFY07}, where one guesses the unknown
device and the task is to minimize the probability of making an error.

On the other hand, in a situation of criticality of errors and very reliable
technology, the perfect retrieving of information is an issue.  Then,
unambiguous discrimination of optical devices \cite{Sed09}, where one allows
for an inconclusive outcome (while, in case of conclusive outcome, the
probability of error is zero) becomes interesting.

In \cite{BDD11} an optimal strategy for the first scenario - namely, the
minimum energy ambiguous discrimination of optical devices - has been provided
for the ideal case. This strategy, that exploits fundamental properties of the
quantum theory such as entanglement, allows for the ambiguous discrimination
of beamsplitters with probability of error under any given threshold, while
minimizing the energy requirement. The proposed optimal strategy has been
compared with a coherent strategy, reminiscent of the one implemented in
common CD readers, showing that the former saves orders of magnitude of energy
if compared with the latter, and moreover allows for perfect discrimination
with finite energy.

In this paper we first extend the results of \cite{BDD11} to the case of
unambiguous ideal quantum reading - namely, the minimum energy unambiguous
discrimination of optical devices. We provide the optimal strategy for
unambiguous discrimination of beamsplitters with probability of failure under
a given threshold, while minimizing the energy requirement. We show that the
optimal strategy does not require any ancillary mode - while in the presence
of noise and loss ancillary states improve the performance of the quantum
reading setup \cite{Sac05, Pir11}. Both strategies for ambiguous and
unambiguous quantum reading reduce to the same optimal strategy for perfect
discrimination if the probability of error (in the former case) or the
probability of failure (in the latter case) is set to zero. Then, we present
some experimental setups implementing such optimal strategies which are
feasible with present day quantum optical technology, in terms of preparations
of single-photon input states, linear optics and photodetectors. Finally, we
will notice that in the experimental implementation of perfect discrimination
the noise is negligible.

There are only a few papers reporting on the experimental implementation of
discrimination of quantum devices. Ref. \cite{ZPWRLHG08} deals with perfect
discrimination between single-bit unitary operations using a sequential
scheme. In Ref.~\cite{LRO09} the authors demonstrate unambiguous
discrimination of non-orthogonal processes employing entanglement. In the
present paper we report on an experimental realization of the perfect
quantum-process discrimination optimized with respect to the minimal energy
flux through the unknown device.

The paper is organized as follows. First, in Section \ref{sect:theory} we
provide general results for the unambiguous discrimination of optical devices
(Section \ref{sect:discdev}) and the optimal strategy for the unambiguous
discrimination of beamsplitters (Section \ref{sect:discbs}). Then, in Section
\ref{sect:setup} we describe some experimental setups for the optimal
ambiguous (Section \ref{sect:ambset}), unambiguous (Section
\ref{sect:unambset}), and perfect (Section \ref{sect:perfset}) discrimination
of beamsplitters. Then, in Section \ref{sect:exper}, we provide the results of
the experimental implementation of the setup for perfect
discrimination. Finally, Section \ref{sect:concl} is devoted to conclusions
and future perspectives.

\section{Unambiguous quantum reading}\label{sect:theory}

A $M$-modes quantum optical device \cite{Leo03} is described by a unitary
operator $U$ relating $M$ input optical modes with annihilation operators
$a_i$ on $\hilb{H}_i$, to $M$ output optical modes with annihilation operators
$a_i'$ on $\hilb{H}_{i'}$, where $\hilb{H}_i$ denotes the Fock space of the
optical mode $i$. We denote the total Fock space as $\hilb{H}=\bigotimes_i
\hilb{H}_i$.

An optical device is called \emph{linear} if the operators of the output modes
are related to the operator of the input modes by a linear transformation,
namely
\begin{align}\label{eq:smatrix}
  \left(\begin{array}{c}{\bf a'}\\{\bf a'}^\dagger\end{array}\right) = S
    \left(\begin{array}{c}{\bf a}\\{\bf a}^\dagger\end{array}\right), \qquad S
      := \left( \begin{array}{cc} A & B \\ \bar{B}&\bar{A} \end{array} \right)
\end{align}
where $S$ is called scattering matrix, $\bar{X}$ denotes the complex conjugate
of $X$, ${\bf a} = (a_1,\dots a_N)$ is the vector of annihilation operators of
the input mode, and analogously ${\bf a'}$ for the output modes. If $B=0$ in
Eq. \eqref{eq:smatrix} the device is called \emph{passive} and conserves the
total number of photons, that is $\bra{\psi} N \ket{\psi} = \bra{\psi}
U^\dagger N U\ket{\psi}$ with $N := \sum_i a_i^\dagger a_i$ the \emph{number
  operator} on $\hilb{H}$. In the following, for any pure state $\ket{\psi}$,
we denote with $\psi := \ket{\psi}\bra{\psi}$ the corresponding projector. For
any Fock space $\hilb{H}$, we denote with $\ket{n}$ a Fock basis in $\hilb{H}$
($\ket{0}$ denotes the state of the vacuum).

Suppose we want to discriminate between two linear optical passive devices
$U_1$ and $U_2$. If a single use of the unknown device is available, the most
general strategy consists of preparing a bipartite input state $\rho \in
\mathcal{B}(\hilb{H} \otimes \hilb{K})$ ($\hilb{K}$ is an ancillary Fock space
with mode operators $b_i$), applying locally the unknown device and performing
a bipartite POVM $\Pi$ on the output state $(\mathcal{U}_x \otimes
\mathcal{I}_{\hilb{K}} ) \rho = (U_x \otimes I_\hilb{K})\rho (U^\dagger_x
\otimes I_\mathcal{K})$ ($x$ can be either $1$ or $2$).
\begin{align}\label{eq:strategy}
  \begin{aligned}
    \Qcircuit @C=0.7em @R=1em { \multiprepareC{1}{\rho} & \ustick{\hilb{H}}
      \qw & \gate{U_x} & \qw & \ghost{\Pi}\\ \pureghost{\rho} &
      \ustick{\hilb{K}} \qw & \qw & \qw & \multimeasureD{-1}{\Pi} }
  \end{aligned}.
\end{align}

The choice of $\Pi$ in Eq. \eqref{eq:strategy} depends on the figure of merit
taken into account. For example, for ambiguous discrimination $\Pi =
\{\Pi_1,\Pi_2\}$ and one's task is to minimize the probability of error
\begin{align}\label{eq:perror}
  P_E(\rho, U_1, U_2) := \Tr[(\mathcal{U}_1 \otimes
    \mathcal{I}_\mathcal{H})(\rho)\Pi_2 + (\mathcal{U}_2 \otimes
    \mathcal{I}_\mathcal{H})(\rho)\Pi_1], \nonumber
\end{align}
with $0 \le P_E(\rho, U_1, U_2) \le 1/2$. When $p_1=p_2=1/2$ the minimal
probability of error has been proven to be given by the following function
\cite{Hel76} of $\rho$,
\begin{align}
  P_E(\rho^*, U_1, U_2) = \frac12 \left( 1-
  ||\left((\mathcal{U}_1-\mathcal{U}_2) \otimes \mathcal{I}_{\hilb{K}} \right)
  \rho ||_1 \right),
\end{align}
where $||X||_1 = \Tr[\sqrt{X^\dagger X}]$ denotes the trace norm.

For unambiguous discrimination $\Pi = \{\Pi_1, \Pi_2, \Pi_I\}$,
$\Tr[(\mathcal{U}_1 \otimes \mathcal{I}_\mathcal{H})(\rho)\Pi_2] =
\Tr[(\mathcal{U}_2 \otimes \mathcal{I}_\mathcal{H})(\rho)\Pi_1] = 0$ and one's
task is to minimize the probability of inconclusive outcome (failure)
\begin{align}\label{eq:pfailure}
  P_F(\rho,U_1,U_2) := \Tr[(\mathcal{U}_1 \otimes \mathcal{I}_\hilb{H} +
    \mathcal{U}_2 \otimes \mathcal{I}_\hilb{H})(\rho)\Pi_I],
\end{align}
with $0 \le P_F(\rho,U_1,U_2) \le 1$.

Upon denoting with $E_D(\rho) := \Tr[\rho (N \otimes I_{\hilb{K}})]$ the
energy that flows through the unknown device, the total energy of the input
state is $E(\rho) := E_D + \Tr[\rho (I_{\hilb{H}} \otimes N_{\hilb{K}})]$.

We introduce now the ambiguous (unambiguous) quantum reading problem. For any
set of two optical devices $\{U_1, U_2\}$ and any threshold $q$ in the
probability of error (failure), find the minimum energy input state $\rho^*$
that allows us to ambiguously (unambiguously) discriminate between $U_1$ and
$U_2$ with probability of error (failure) not greater than $q$, namely
\begin{align}\label{eq:figmer1}
  \rho^* = \arg \min_{\rho \textrm{ s.t. } P(\rho, U) \leq q} E(\rho).
\end{align}
where $P(\rho,U)=P_E(\rho,U)$ for the ambiguous discrimination problem and
$P(\rho,U)=P_F(\rho,U)$ for the unambiguous discrimination problem.

\subsection{Unambiguous quantum reading of optical devices}\label{sect:discdev}

The problem in Eq. \eqref{eq:figmer1} has been already solved in \cite{BDD11}
for the case of ambiguous discrimination. Here we generalize the results
obtained in \cite{BDD11} to the unambiguous discrimination problem.

First, notice that for any POVM $\Pi$ we have $P_F((\mathcal{U}_1 \otimes
\mathcal{I}_{\hilb{K}})\rho, I, U_2U_1^\dagger ) = P_F(\rho, U_1, U_2)$ and
$E((\mathcal{U}_1 \otimes I_{\hilb{K}})\rho) = E(\rho)$, so we can restrict
our analysis to the case in which $U_1 = I$ and $U_2 = U$, and identify
$P_F(\rho,I,U) = P_F(\rho,U)$.

Then, notice that without loss of generality the constraint in
Eq. \eqref{eq:figmer1} can be restated as $P_F(\rho,U) = q$. Indeed, for any
POVM $\Pi$ we have that $P_F(\rho,U)$ is a continuous function, and that
$P_F(\ketbra{0}{0},U) = 1$. So for any $\rho$ with $P_F(\rho,U) < q$ there
exists a $0 < \alpha \leq 1$ such that $P_F((1 - \alpha )\rho + \alpha
\ketbra{0}{0},U) = q$. Since $E((1 - \alpha )\rho + \alpha \ketbra{0}{0}) <
E(\rho)$, the constraint in Eq. \eqref{eq:figmer1} becomes $P_F(\rho,U) = q$.

\begin{thm}[Optimal state is pure]\label{thm:pure}
  For any optical device $U$ and any threshold $q$ in the probability of
  failure $P_F(\rho,U)$, there exists a state $\rho^*$ which minimizes
  Eq. \eqref{eq:figmer1} such that $\rho^*$ is pure.
\end{thm}

\begin{proof}
  Notice that Eq. \eqref{eq:figmer1} is equivalent to $C(\rho,U) := p
  P_F(\rho,U) + (1-p) E(\rho)$, for any fixed value of $p$. If $\rho^*$ is the
  state that minimizes $C(\rho,U)$, for $q := P(\rho^*,U)$ we have that
  $E(\rho^*)$ gives the minimum possible value for the energy. Since
  $P_F(\rho,U)$ and $E(\rho)$ are linear functions of $\rho$, it follows that
  $C(\rho,U)$ is a linear function of $\rho$ and its minimum is attained on
  the boundary of its domain, namely for a pure state $\ket{\psi^*}$.
\end{proof}

As a consequence of Proposition \ref{thm:pure}, Eq. \eqref{eq:figmer1} can be
restated as
\begin{align}\label{eq:figmer3}
  \psi^* = \arg\min_{\psi \textrm{ s.t. } P(\psi,U) = q} E(\psi).
\end{align}
For pure states, the probability of failure in the unambiguous discrimination
when $p_1=p_2=1/2$ given by Eq. \eqref{eq:pfailure} has been proved to be
given by \cite{Sed09}
\begin{align}\label{eq:pfailure2}
  P_F(\psi^*,U) & = |\bra{\psi} U \ket{\psi}|.
\end{align}

\begin{thm}[No ancillary modes are required]\label{thm:discdev}
  For any optical device $U$ and any threshold $q$ in the probability of
  failure $P_F(\rho,U)$, there exists a state $\rho^*$ which minimizes
  Eq. \eqref{eq:figmer1} such that $\rho^* \in \hilb{H}$.
\end{thm}

\begin{proof}
  Any pure input state can be written as $\ket{\psi} = \sum_i c_i
  \ket{i}\ket{\chi_i}$ where $\ket{i}$ is an orthonormal basis in
  $\mathcal{H}$ and $\ket{\chi_i}$ are normalized states in $\hilb{K}$. If we
  define $\ket{{\psi}'} := \sum_i c_i \ket{i}\ket{0}$, it follows that
  $P_F(\psi,U) = P_F(\psi',U)$ while $E(\psi) \geq E({\psi}')$, which proofs
  of the statement.
\end{proof}

Since no ancillary modes are required, the energy $E_D(\psi)$ that flows
through the unknown device is equal to the total energy of the input state
$E(\rho)$, so minimizing the former instead than the latter - namely,
replacing $E(\psi)$ with $E_D(\psi)$ in Eq. \eqref{eq:figmer3} - does not
change the optimal state.

Notice that the generalization of the results in \cite{BDD11} provided here
basically depends on some common properties of the probability of error in
Eq. \eqref{eq:perror} (for ambiguous discrimination) and the probability of
failure in Eq. \eqref{eq:pfailure} (for unambiguous discrimination), namely
the linearity in $\rho$, the equalities $P_E(\ketbra{0}{0}) =
P_F(\ketbra{0}{0}) = 0$ [Eq. \eqref{eq:pfailure}], and the monotonicity in
$|\bra{\psi} U \ket{\psi}|$ (Eq. \eqref{eq:pfailure2}).

\subsection{Unambiguous quantum reading of beamsplitters}\label{sect:discbs}

A beamsplitter is a two-mode linear passive quantum optical device such that
$A \in SU(2)$ in Eq. \eqref{eq:smatrix}. In the following we will use the
basis $\{ \ket{n,m} \}$ with respect to which $A$ is diagonal with eigenvalues
$e^{\pm i \delta}$, $0 \le \delta \le \pi$. With this choice, for any
$\ket{\psi} = \sum_{n,m=0}^\infty \alpha_{n,m} \ket{n,m}$, we have $U
\ket{n,m} = e^{i\delta (n-m)}\ket{n,m}$, so that $\bra{\psi} U \ket{\psi} =
\sum_{n,m=0}^\infty |\alpha_{n,m}|^2 e^{i\delta(n-m)}$ and
$\bra{\psi}N\ket{\psi} = \sum_{n,m=0}^\infty |\alpha_{n,m}|^2 (n+m)$. We
notice that both these expressions only depends on the squared modulus of the
coefficients $\alpha_{n,m}$, so we can assume $\alpha_{n,m}$ to be real and
positive.

Here $\floor{x}$ ($\ceil{x}$) denotes the maximum (minimum) integer number
smaller (greater) than $x$.
\begin{thm}[Unambiguous quantum reading of beamsplitters]\label{thm:discbs}
  For any beamsplitter $U$ and for any threshold $q$ in the probability of
  failure, there exists a state $\psi^*$ which minimizes
  Eq. \eqref{eq:figmer3} such that
  \begin{align}
    \ket{\psi^*} = \frac{1}{\sqrt2} \alpha (\ket{0,n^*} + \ket{n^*,0}) +
    \beta \ket{00},
  \end{align}
  where $|\alpha| = \sqrt{\frac{1-q}{1-\cos(\delta n_1)}}$, $|\beta| =
  \sqrt{1-|\alpha|^2}$, $n^* = \arg \min_{\floor{x^*},\ceil{x^*}} E(\psi^*)$,
  and $x^* = \min(x \in \mathbb{R}^+ | \delta x = \tan (\delta x/2))$.
\end{thm}

\begin{proof}
  First we prove that the optimal state in Eq. \ref{eq:figmer3} is a
  superposition of NOON states.  For any state $\ket{\psi} = \sum_{n,m}
  \alpha_{n,m} \ket{n,m}$, the state $\ket{\psi'} = \sqrt{1/2} \sum_l
  \alpha'_{l} (\ket{l,0}+\ket{0,l})$ with $|\alpha'_l|^2 = \sum_{|n-m|=l}
  |\alpha_{nm}|^2$ is such that
  \begin{align}
    \bra{\psi'}N\ket{\psi'} & = \sum_{n,m=0}^\infty \alpha_{nm}^2 |n-m| \le
    \bra{\psi}N\ket{\psi},\\ | \bra{\psi'}U\ket{\psi'} | & = \left |
    \sum_{n,m=0}^\infty \alpha_{nm}^2 \cos(\delta|n-m|) \right| \le
    |\bra{\psi}U\ket{\psi} |. \nonumber
  \end{align}
  So we have $\bra{\psi}U\ket{\psi} \in \mathbb{R}$ and the constraint in
  Eq. \eqref{eq:figmer3} becomes $\bra{\psi}U\ket{\psi} = q$.

  Then we prove that the optimal state is the superposition of two NOON
  states. Let $\ket{\psi^*} = \sqrt{1/2} \sum_n \alpha^*_{n}
  (\ket{n,0}+\ket{0,n})$ be the optimal state and let the set $\{\alpha_{n}^*
  \}$ have $N \geq 3$ not-null elements. Then there exist $n_1$ and $n_2$ such
  that $\alpha_{n_1},\alpha_{n_2}\neq 0$ and $\cos(\delta n_1) \leq q \leq
  \cos(\delta n_2)$. Define $\ket{\chi} := 1/\sqrt2 \sum_{i=1,2} \beta_{n_i}
  (\ket{n_i,0} + \ket{0,n_i})$ such that $\bra{\chi} U \ket{\chi} = q$, and
  $\ket{\xi} := 1/\sqrt2 (1-\epsilon)^{-1/2} \sum_n \gamma_n (\ket{n,0} +
  \ket{0,n})$, where
  \begin{align*}
   \gamma_n = \left\{ \begin{array}{ll} \alpha_n & \mbox{ if } n \ne
        n_1,n_2\\ \sqrt{\alpha_n^2 - \epsilon \beta_n^2} & \mbox{ if } n =
        n_1,n_2  \end{array} \right. ,
  \end{align*}
  and $\epsilon \le \min(\alpha_{n_1}/\beta_{n_1},\alpha_{n_2}/\beta_{n_2})$.
  Notice that $\bra{\xi}U\ket{\xi}=q$, and $\bra{\psi^*} N \ket{\psi^*} =
  \epsilon \bra{\chi}N\ket{\chi} + (1-\epsilon) \bra{\xi}N\ket{\xi}$. If
  $\bra{\chi} N \ket{\chi} = \bra{\psi^*} N \ket{\psi^*}$ the statement
  follows with $\ket{\psi} = \ket{\chi}$. If $\bra{\chi}N\ket{\chi} \neq
  \bra{\psi^*} N \ket{\psi^*}$, either $\bra{\chi}N\ket{\chi} < \bra{\psi^*} N
  \ket{\psi^*}$ or $\bra{\xi}N\ket{\xi} < \bra{\psi^*} N \ket{\psi^*}$, that
  contradicts the hypothesis that $\ket{\psi^*}$ is the optimal state.

  Finally we prove that the optimal state is the superposition of a NOON state
  and the vacuum. Let $\ket{\psi^*} = 1/\sqrt2 \sum_{i=1,2} \alpha_{n_i}
  (\ket{n_i,0} + \ket{0,n_i})$. Then
  \begin{align*}
    \bra{\psi^*} N \ket{\psi^*} = \frac{n_2 \cos(\delta n_1) - n_1\cos(\delta
      n_2) + q(n_1-n_2)}{\cos(\delta n_1)-\cos(\delta n_2)}.
  \end{align*}
  It is easy to verify (in \cite{BDD11} a proof of this fact is
  provided) that it is not restrictive to set $n_2=0$, so one has
  $\bra{\psi^*} N \ket{\psi^*} = (1-q) n (1-\cos(\delta n_1))^{-1}$.  Then one
  can see that it is not restrictive to choose $\pi/2 \le \delta n_1 \le \pi$,
  where $\bra{\psi^*} N \ket{\psi^*}$ can be proven to be a convex function
  that attains its minimum for $n_1 = \floor{x^*}, \ceil{x^*}$, with $x^* =
  \min(x \in \mathbb{R}^+ | \delta x = \tan (\delta x/2))$. The statement
  immediately follows.
\end{proof}

Notice that from Proposition \ref{thm:discbs} it immediately follows that
unambiguous discrimination between beamsplitters $U$ and $I$ can be achieved
only if the threshold $q$ in the probability of failure $P_F(\rho,U)$ satisfies
the inequality $q \geq \cos(\delta n^*)$ (an analogous inequality, namely
$K(q) \geq \cos(\delta n^*)$ with $K(q) = \sqrt{4q(1-q)}$, holds in the case
of ambiguous quantum reading addressed in \cite{BDD11}).

In \cite{BDD11}, the optimal coherent strategy - namely, a strategy making use
of coherent input states and homodyne measurements - for ambiguous quantum
reading is provided. A comparison between the optimal strategy and the optimal
coherent strategy showed that the former requires much less energy than the
latter when the same threshold in the probability of error is allowed. Here, we
notice that, since the support of coherent states is the entire Fock space, no
measurement exists projecting on its orthogonal complement. For this reason,
no coherent strategy exists for the unambiguous discrimination of optical
devices.

\section{Experimental setup for quantum reading}
\label{sect:setup}

In this Section we provide experimental setups for ambiguous, unambiguous, and
perfect quantum reading, which are feasible with present quantum optical
technology. The input is a single-photon state, that can be realized
e.g. through spontaneous parametric down conversion or through the attenuation
of a laser beam. The evolution is given by a circuit of beamsplitters, one of
which is the unknown one, and the final measurement is implemented through
photodetectors.

In Proposition \ref{thm:discdev} we proved that, for the unambiguous quantum
reading of optical devices, no ancillary modes are required. The same result
has been proved for the case of ambiguous quantum reading in
\cite{BDD11}. Nevertheless, the proposed setups for quantum reading make use
of three-modes input states - namely, an ancillary mode is employed. This
choice is due to the requirement to have an input state with fixed number of
photons in order to be able to take into account loss. For this reason, our
setup minimizes the energy $E_D(\rho)$ that flows through the unknown device,
while the total energy of the input state is fixed.

In the following, for any beamsplitter $X$ we denote with $A_X$ the $A$ matrix
of $X$ in Eq. \eqref{eq:smatrix}, so we write
\begin{align*}
  A_X = \left( \begin{array}{cc} r_X & -t_X \\ t_X & r_X \end{array} \right),
  \qquad A_X^\dagger = \left( \begin{array}{cc} r_X & t_X \\ -t_X &
    r_X \end{array} \right).
\end{align*}
We define the reflectivity $R_X$ and the transmittivity $T_X$ of $X$ as $R_X
:= |r_X|^2$ and $T_X = |t_X|^2$, respectively, with $R_X + T_X = 1$.

The general setup is given by a Mach-Zender interferometer with beamsplitters
$B$ and $B^\dagger$, acting on modes $2$ and $3$. In one of the harms of the
interferometer (corresponding to mode $2$), the following beamsplitters are
inserted
\begin{align*}
  \begin{aligned}
    \Qcircuit @C=2mm @R=1em { \prepareC{0} & \ustick{1} \qw &\qw &
      \multigate{1}{N} & \multigate{1}{D} & \multigate{1}{I,U} &
      \multigate{1}{D^\dagger} & \multigate{1}{N^\dagger} & \qw &
      \multimeasureD{2}{\Pi}\\ \prepareC{0} & \ustick{2} \qw &
      \multigate{1}{B} & \ghost{N} & \ghost{D} & \ghost{I,U} &
      \ghost{D^\dagger} & \ghost{N^\dagger} & \multigate{1}{B^\dagger} &
      \ghost{\Pi}\\ \prepareC{1} & \ustick{3} \qw & \ghost{B} & \qw & \qw &
      \qw & \qw & \qw & \ghost{B^\dagger} & \ghost{\Pi} }
  \end{aligned},
\end{align*}
where $N$ is a $50/50$ beamsplitter, $I, U$ is the unknown beamsplitter, and
$D$ is the beamsplitter diagonalizing $U$. The POVM $\Pi$ is different for
ambiguous and unambiguous quantum reading.

It is easy to verify that the composition of beamsplitters $DN$ reduces to a
phase shifter on mode $2$, namely
\begin{align}
  A_D = \frac1{\sqrt{2}} \left( \begin{array}{cc} 1 & 1 \\ i & -i \end{array}
  \right), \qquad A_D A_N = \left( \begin{array}{cc} 1 & 0 \\ 0 &
    i \end{array} \right).
\end{align}
It is easy to check that this phase shifter is irrelevant, so in the following
we will disregard it.

\subsection{Experimental setup for ambiguous quantum reading of beamsplitters}
\label{sect:ambset}

The optimal strategy for ambiguous quantum reading of beamsplitters has been
provided in \cite{BDD11}. Here we describe an experimental setup implementing
such strategy, namely the ambiguous discrimination of a beamsplitter randomly
chosen from the set $\{I, U\}$ with equal prior probabilities, with
probability of error $P_E(\rho,U)$ under a given threshold $q$ and minimal
energy flow trough the unknown device. In the following we set $K(q) :=
\sqrt{4q(1-q)}$.  According to \cite{BDD11}, in order to have $P_E(\rho, U)
\le q$, we must have $K(q) \geq \sqrt{R_U}$.

The experimental setup is then given by
\begin{align*}
  \begin{aligned}
    \Qcircuit @C=1em @R=1em { \prepareC{0} & \ustick{1} \qw & \qw &
      \multigate{1}{I,U} & \qw & \multigate{1}{M} & \qw &
      \measureD{I}\\ \prepareC{0} & \ustick{2} \qw & \multigate{1}{B} &
      \ghost{I,U} & \multigate{1}{B^\dagger} & \ghost{M} &
      \multigate{1}{N^\dagger} & \measureD{\Pi_U}\\ \prepareC{1} & \ustick{3}
      \qw & \ghost{B} & \qw & \ghost{B^\dagger} & \qw & \ghost{N^\dagger} &
      \measureD{\Pi_I} \gategroup{1}{6}{3}{8}{3mm}{--}}
  \end{aligned},
\end{align*}
where the reflectivities and transmittivities of beamsplitters $B$, $M$ and
$N^{\dagger}$ are given by
\begin{align*}
  R_B = \frac{K(q)-r_U}{1-r_U}, \qquad T_B = \frac{1-K(q)}{1-r_U},\\
  R_M = \frac{(1-K(q))(K(q)-r_U)}{(1-2q)^2},\\
  T_M = \frac{(1-K(q))(1+r_U)}{(1-2q)^2},\\
  R_N = \sqrt{1-q}, \qquad T_N = \sqrt{q}.
\end{align*}

The optimal measurement for ambiguous discrimination \cite{Hel76} is
implemented by the two beamsplitters $M$ and $N^\dagger$ and by the two
photocounters $\Pi_U$ and $\Pi_I$ surrounded by the dashed line (no
measurement is performed on output mode $1$). The conditional probabilities
$p_{X|Y}$ of detecting a photon in photodetector $\Pi_X$ given that the
unknown device is $Y$ are given by
\begin{align*}
  p_{U|U} = p_{I|I} = 1-q, \qquad p_{I|U} = p_{U|I} = q.
\end{align*}
Detecting a photon in $\Pi_U$ or $\Pi_I$ implies that the unknown beamsplitter
is $U$ or $I$, respectively, with probability of error $q$.

\subsection{Experimental setup for unambiguous quantum reading of
  beamsplitters}
\label{sect:unambset}

We provided the optimal strategy for unambiguous quantum reading of
beamsplitters in Proposition \ref{thm:discbs}. Here we describe an
experimental setup implementing such strategy, namely the unambiguous
discrimination of a beamsplitter randomly chosen from the set $\{I, U\}$ with
equal prior probabilities, with probability of failure $P_F(\rho,U)$ under a
given threshold $q$ and minimal energy flow trough the unknown
device. According to Proposition \ref{thm:discbs}, in order to have $P_F(\rho,
U) \le q$, we must have $q \geq \sqrt{R_U}$.

The experimental setup is given by
\begin{align*}
  \begin{aligned}
    \Qcircuit @C=1em @R=1em { \prepareC{0} & \ustick{1} \qw & \qw &
      \multigate{1}{I,U} & \qw & \multigate{1}{M} & \qw & \measureD{\Pi_U}\\
      \prepareC{0} & \ustick{2}
      \qw & \multigate{1}{B} & \ghost{I,U} & \multigate{1}{B^\dagger} &
      \ghost{M} & \multigate{1}{N} & \measureD{\Pi_F}\\
      \prepareC{1} & \ustick{3} \qw & \ghost{B} & \qw &
      \ghost{B^\dagger} & \qw & \ghost{N} & \measureD{\Pi_I}
    \gategroup{1}{6}{3}{8}{3mm}{--}}
  \end{aligned},
\end{align*}
where the reflectivities and transmittivities of beamsplitters $B$, $M$ and
$N$ are given by
\begin{align*}
  R_B = \frac{q-r_U}{1-r_U}, \qquad T_B = \frac{1-q}{1-r_U},\\
  R_M = \frac{(\sqrt{1+r_U} - \sqrt{q(q-r_U)})^2}{(1+q)^2},\\
  T_M = \frac{(\sqrt{q(1+r_U)}+\sqrt{q-r_U})^2}{(1+q)^2},\\
  R_N = \sqrt{1-q}, \qquad T_N =\sqrt{q}.
\end{align*}

The optimal measurement for unambiguous discrimination \cite{Sed09} is
implemented by the two beamsplitters $M$ and $N$ and by the three
photocounters $\Pi_U$, $\Pi_I$, and $\Pi_F$ surrounded by the dashed line.
The conditional probabilities $p_{X|Y}$ of detecting a photon in
photodetector $\Pi_X$ given that the unknown device is $Y$ are given by
\begin{align*}
  \begin{aligned}
    & p_{U|U} = p_{I|I} = 1-q, \quad p_{I|U} = p_{U|I} = 0, \\
& p_{F|U} =  p_{F|I} = q.
  \end{aligned}
\end{align*}
Detecting a photon in $\Pi_U$ or $\Pi_I$ implies that the unknown beamsplitter
is certainly $U$ or $I$, respectively, while detecting a photon in $\Pi_F$
declares a failure with probability $q$.

\subsection{Experimental setup for perfect quantum reading of beamsplitters}
\label{sect:perfset}

The optimal strategies for ambiguous and unambiguous quantum reading of
beamsplitters of \cite{BDD11} and Proposition \ref{thm:discbs} reduce to the
same optimal strategy for perfect quantum reading of beamsplitters when the
threshold $q$ in the probability of error (for the ambiguous case) and failure
(for the ambiguous case) is set to zero.

In this case, the condition $r_U \le 0$ given by Proposition \ref{thm:discbs}
can be satisfied upon defining $U$ as the composition of a beamsplitter $V$
with $r_V \ge 0$ with $\pm \pi/2$ phase shifters on its input and output
modes, according to
\begin{align}\label{eq:uperf}
  \begin{aligned}
    \Qcircuit @C=1em @R=1em { & \ustick{1} & \multigate{1}{U} & \qw\\ &
      \ustick{2} & \ghost{U} & \qw }
  \end{aligned}
  :=
  \begin{aligned}
    \Qcircuit @C=1em @R=1em { & \ustick{1} & \gate{-\pi/2} & \multigate{1}{V}
      & \gate{-\pi/2} & \qw\\ & \ustick{2} & \gate{\pi/2} & \ghost{V} &
      \gate{\pi/2} & \qw }
  \end{aligned},
\end{align}
and it is easy to verify that
\begin{align*}
  A_U = \left( \begin{array}{cc} e^{-i\frac{\pi}2} & 0 \\ 0 & e^{i\frac{\pi}2}
    \end{array} \right) A_V \left( \begin{array}{cc} e^{-i\frac{\pi}2} & 0
    \\ 0 & e^{i\frac{\pi}2} \end{array} \right) = \left( \begin{array}{cc}
    -r_V & -t_V \\ t_V & -r_V \end{array} \right),
\end{align*}
with $0 \le r_V,t_V \le 1$, so that $r_U \le 0$ and perfect discrimination is
indeed possible.

The experimental setup is then given by
\begin{align}\label{eq:perfsetup}
  \begin{aligned}
    \Qcircuit @C=1em @R=1em { \prepareC{0} & \ustick{1} \qw & \qw &
      \multigate{1}{I,U} & \qw & \measureD{\Pi_U}\\ \prepareC{0} & \ustick{2}
      \qw & \multigate{1}{B} & \ghost{I,U} & \multigate{1}{B^\dagger} &
      \measureD{\Pi_{U'}}\\ \prepareC{1} & \ustick{3} \qw & \ghost{B} & \qw &
      \ghost{B^\dagger} & \measureD{\Pi_I} }
  \end{aligned},
\end{align}
where the reflectivity and transmittivity of beamsplitter $B$ are given by
\begin{align*}
  R_B = \frac{r_V}{1+r_V}, \qquad T_B = \frac{1}{1+r_V}.
\end{align*}

Notice that the two $-\pi/2$ phase shifters on mode $1$ in the Equation
\eqref{eq:uperf} are irrelevant and can be discarded, since the one on the
input mode acts on the vacuum and the one on the output mode is immediately
followed by a photodetector, so we can redefine
\begin{align}\label{eq:BS}
  \begin{aligned}
    \Qcircuit @C=1em @R=1em { & \ustick{1} & \multigate{1}{U} & \qw\\ &
      \ustick{2} & \ghost{U} & \qw }
  \end{aligned}
  :=
  \begin{aligned}
    \Qcircuit @C=1em @R=1em { & \ustick{1} & \qw & \multigate{1}{V} & \qw &
      \qw\\ & \ustick{2} & \gate{\pi/2} & \ghost{V} & \gate{\pi/2} & \qw }
  \end{aligned}.
\end{align}

The optimal measurement for perfect discrimination is implemented by the three
photocounters $\Pi_U$, $\Pi_{U'}$, and $\Pi_I$. The conditional probabilities
$p_{X|Y}$ of detecting a photon in photodetector $\Pi_X$ given that the
unknown device is $Y$ are given by
\begin{align*}
  p_{U|U} = 1 - r_V, \qquad p_{U'|U} = r_V, \qquad p_{I|I} = 1,\\
  p_{I|U} = p_{U|I} = p_{U'|I} = 0
\end{align*}
Detecting a photon in $\Pi_U$ or $\Pi_{U'}$ implies that the unknown beamsplitter
is certainly $U$, while detecting a photon in $\Pi_I$ implies that the unknown
device is certainly $I$.

\section{Experimental implementation of quantum reading}
\label{sect:exper}

To demonstrate experimental feasibility of quantum reading we have built a
laboratory setup for perfect discrimination of two beam splitters according to
scheme \eqref{eq:perfsetup} [see Fig. (\ref{fig-scheme})]. It consists of a
Mach-Zehnder interferometer (MZI) with an additional beam splitter in its
upper arm. This additional beam splitter has a variable splitting ratio and it
serves as an unknown device to be discriminated.

\begin{figure}[hptb]
  \includegraphics[scale=.35]{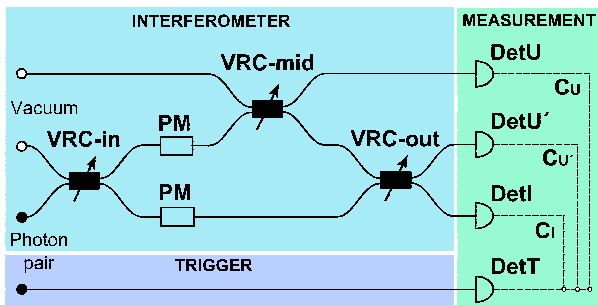}
  \caption{(Color online) Scheme of the experiment. VRC -- variable ratio
    couplers, PM -- phase modulators, Det -- detectors, C -- coincidence
    electronics.}
  \label{fig-scheme}
\end{figure}

We use a heralded single photon source based on spontaneous parametric down
conversion (SPDC). Namely, we employ a collinear frequency-degenerate SPDC
process with type-II phase matching in a 2-mm-long BBO crystal pumped by a cw
laser diode (Coherent Cube) at 405\,nm. In this process pairs of photons at
810\,nm are created. Photons from each pair are separated by a polarizing beam
splitter and coupled into single-mode optical fibers. One of them is led
directly to trigger detector DetT which heralds the creation of a pair
(Perkin-Elmer SPCM AQR-14FC, dark counts 180\,s$^{-1}$, total efficiency cca
50\%). The other one enters MZI through variable ratio coupler VRC-in.

An additional variable ratio coupler, VRC-mid, represents an unknown
device. When its reflectivity equals one it corresponds to device I. To switch
to device U one has to set required splitting ratio and apply additional phase
shift, see scheme \eqref{eq:BS}. For practical reasons we apply---without the
loss of generality---a cumulative phase shift in front of the beam
splitter. In the experiment the phase shifts are introduced by electro-optical
phase modulators (PM) manufactured by EO Space. Their half-wave voltages are
about 1.5\,V. These phase modulators exhibit relatively high
dispersion. Therefore one PM is placed in each interferometer arm in order to
compensate dispersion effects. In case of device U we use the PM in the upper
interferometer arm to apply the additional phase shift of $\pi$.

The output fibers from the unknown device and from the interferometer are led
to detectors DetU, DetU', and DetI. These detectors are parts of Perkin-Elmer
quad module SPCM-AQ4C (dark counts 370-440\,s$^{-1}$, total efficiencies about
50\%).

To reduce the effect of the phase drift caused by fluctuations of
temperature and temperature gradients we apply both passive and active
stabilization. The experimental setup is covered by a shield minimizing air flux around the
components. Besides, after each three seconds of measurement an active
stabilization is performed. It measures intensities for phase shifts 0 and
$\pi/2$ and if necessary it calculates phase compensation and applies
corrective voltage to the phase modulator in the lower interferometer
arm. This results in the precision of the phase setting during the measurement
period better than $\pi/200$.

For each pair of devices U and I the proper splitting ratio of fiber couplers
VRC-in and VRC-out must be set in order to discriminate these devices
optimally. We have made measurements for 11 different devices U with intensity
reflectances $0, 0.1, 0.2,\dots, 1$. For each pair of devices U and I the
counts at detectors DetU, DetU', and DetI were cumulated during 30
three-second measurement intervals interlaced by stabilization procedures.
All measurements were done in coincidence with the trigger detector DetT. It
means we measured coincidence counts $C_U, C_{U'}, C_I$ between detectors
DetT-DetU, DetT-DetU', and DetT-DetI, respectively, using 3\,ns coincidence
time window. These results were normalized to obtain relative frequencies,
$f_j = C_j/(C_U+C_{U'}+C_I), j=U, U', I$, which can be compared with
theoretical probabilities of detection.

Measured relative frequencies and theoretical probabilities are listed in
Tables \ref{tab:deviceI} and \ref{tab:deviceU} and shown in Figures
\ref{deviceI} and \ref{deviceU}, respectively. Table \ref{tab:deviceI} and
Fig. \ref{deviceI} show the results obtained with device I inserted, Table
\ref{tab:deviceU} and Fig. \ref{deviceU} summarize the results for devices
U. Each row in the tables corresponds to one pair of I and U with $R_v$ being
the reflectivity of device U. One can observe very good agreement between
theory and experiment. Small discrepancies appear mainly due to imperfections
in splitting-ratio settings, phase fluctuations, and polarization
misalignment. In coincidence measurements the contribution of detector noise
is completely negligible.

\begin{table}[hptb]
  \begin{center}
    \begin{tabular}{|c|ccc|ccc|}
      \hline
      $R_v$ & $p_{U|I}$ & $p_{U'|I}$ & $p_{I|I}$ & $f_{U|I}$ & $f_{U'|I}$ & $f_{I|I}$ \\
      \hline \hline
      0.0   & 0  & 0  & 1  & 0.000  & 0.002  & 0.998   \\
      0.1   & 0  & 0  & 1  & 0.000  & 0.012  & 0.988   \\
      0.2   & 0  & 0  & 1  & 0.000  & 0.018  & 0.982   \\
      0.3   & 0  & 0  & 1  & 0.000  & 0.012  & 0.988   \\
      0.4   & 0  & 0  & 1  & 0.000  & 0.023  & 0.977   \\
      0.5   & 0  & 0  & 1  & 0.000  & 0.022  & 0.978   \\
      0.6   & 0  & 0  & 1  & 0.000  & 0.014  & 0.986   \\
      0.7   & 0  & 0  & 1  & 0.000  & 0.011  & 0.989   \\
      0.8   & 0  & 0  & 1  & 0.000  & 0.013  & 0.987   \\
      0.9   & 0  & 0  & 1  & 0.000  & 0.018  & 0.982   \\
      1.0   & 0  & 0  & 1  & 0.000  & 0.021  & 0.979   \\
      \hline
    \end{tabular}
    \caption{\label{tab:deviceI} Results for device I. $R_v$ -- reflectivities
      of devices U; $p_{U|I}, p_{U'|I}, p_{I|I}$ -- theoretical probabilities
      of photon detection at detectors DetU, DetU', DetI, respectively,
      $f_{U|I}, f_{U'|I}, f_{I|I}$ -- relative frequencies measured at
      detectors DetU, DetU', DetI, respectively (measured in coincidence with
      DetT).}
  \end{center}
\end{table}

\begin{figure}[hptb]
  \includegraphics[scale=.4]{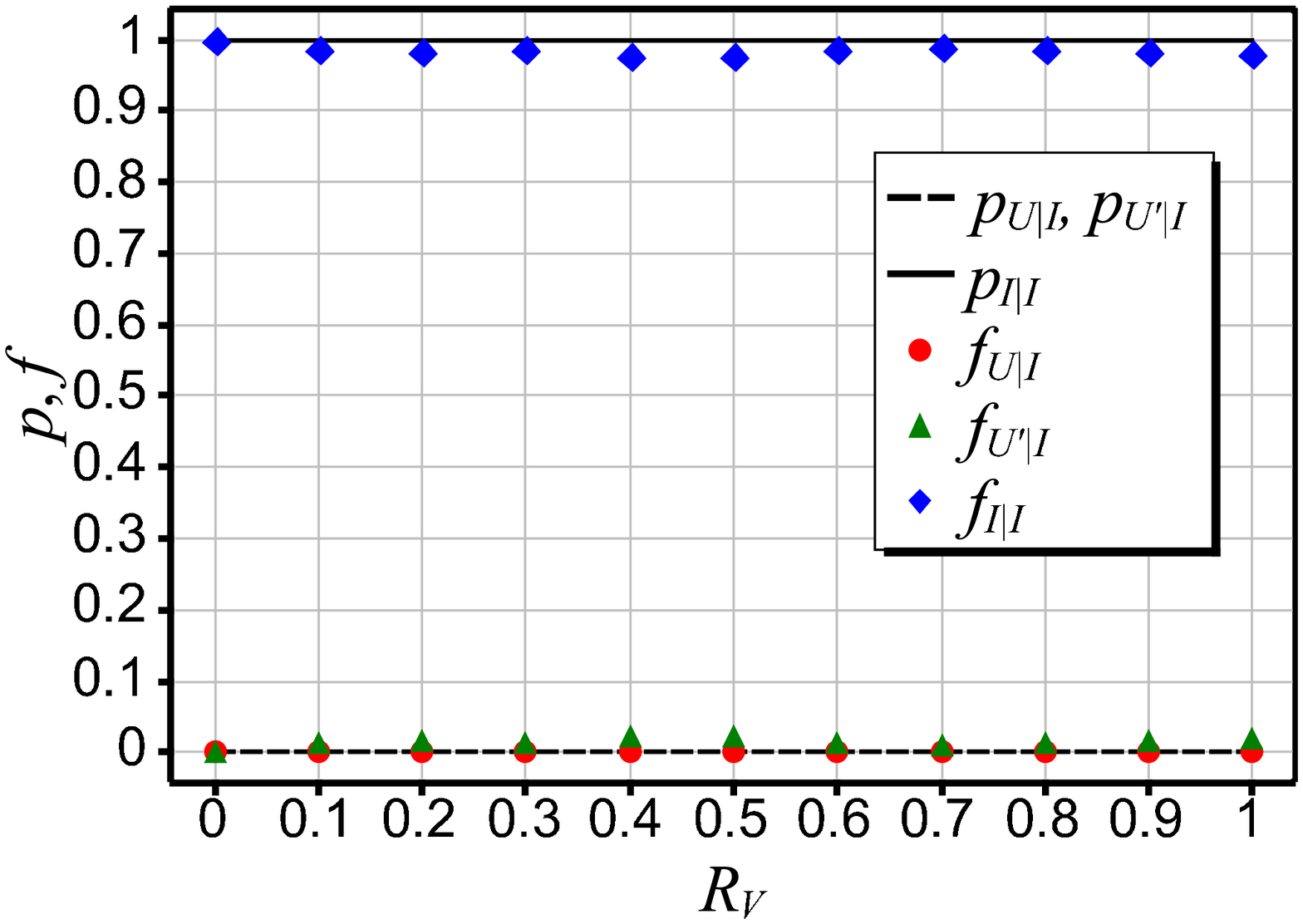}
  \caption{(Color online) Results for device I: Detection probabilities and
    measured relative frequencies as functions of reflectivity $R_v$. Different
    reflectivities correspond to different devices U (VRC-mid).}
  \label{deviceI}
\end{figure}

\begin{table}[hptb]
  \begin{center}
    \begin{tabular}{|c|ccc|ccc|}
      \hline
      $R_v$ & $p_{U|U}$ & $p_{U'|U}$ & $p_{I|U}$ & $f_{U|U}$ & $f_{U'|U}$ & $f_{I|U}$\\
      \hline \hline
      0.0  &  1.000  &  0.000      &   0   & 0.986  &  0.000   & 0.014 \\
      0.1  &  0.684  &  0.316      &   0   & 0.680  &  0.295   & 0.025 \\
      0.2  &  0.553  &  0.447      &   0   & 0.551  &  0.440   & 0.009 \\
      0.3  &  0.452  &  0.548      &   0   & 0.455  &  0.542   & 0.003 \\
      0.4  &  0.368  &  0.633      &   0   & 0.369  &  0.623   & 0.008 \\
      0.5  &  0.293  &  0.707      &   0   & 0.288  &  0.691   & 0.021 \\
      0.6  &  0.225  &  0.775      &   0   & 0.219  &  0.758   & 0.022 \\
      0.7  &  0.163  &  0.837      &   0   & 0.160  &  0.830   & 0.010 \\
      0.8  &  0.106  &  0.894      &   0   & 0.100  &  0.891   & 0.009 \\
      0.9  &  0.051  &  0.949      &   0   & 0.046  &  0.946   & 0.007 \\
      1.0  &  0.000  &  1.000      &   0   & 0.000  &  0.980   & 0.020 \\
      \hline
    \end{tabular}
    \caption{\label{tab:deviceU} Results for devices U. $R_v$ -- reflectivity
      of device U; $p_{U|U}, p_{U'|U}, p_{I|U}$ -- theoretical probabilities
      of photon detection at detectors DetU, DetU', DetI, respectively,
      $f_{U|U}, f_{U'|U}, f_{I|U}$ -- relative frequencies measured at
      detectors DetU, DetU', DetI, respectively (measured in coincidence with
      DetT).}
  \end{center}
\end{table}

\begin{figure}[hptb]
  \includegraphics[scale=0.4]{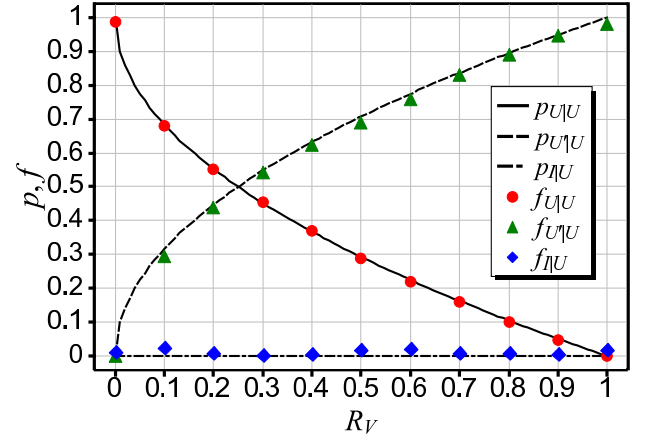}
  \caption{(Color online) Results for devices U: Detection probabilities and
    measured relative frequencies as functions of reflectivity $R_v$. Different
    reflectivities correspond to different devices U (VRC-mid).}
  \label{deviceU}
\end{figure}

\section{Conclusion}
\label{sect:concl}

In this paper we considered unambiguous quantum reading of optical memories,
on the assumption that noise and loss are negligible. In Section
\ref{sect:theory} we showed that the optimal strategy for the unambiguous
discrimination of optical devices can be derived by extending the results
proved for the ambiguous case.

In Section \ref{sect:setup} we presented some experimental implementation of
quantum reading for both the ambiguous and the unambiguous case. In the
proposed setups the input state is fixed to be a single photon state. By
making use of an ancillary mode it was possible to tune the amount of energy
flowing through the device.

Finally, in Section \ref{sect:exper} we provide experimental results for the
perfect quantum reading. The advantage of the implemented setup is that in
ideal case there is exactly one photon at the output ports. It makes detection
relatively easy. Nevertheless, it is still a superposition of a single photon
and vacuum what is entering the unknown device. So the unknown device is
exposed just to a fraction of energy of a single photon in average. Even if
the overall probability of success of the setup is relatively low because of
technological losses, we were able to measure precisely the relative
probabilities of all outputs and our experiment convincingly validate the
predictions of the exposed theory.

\section*{Acknowledgments}

We thank Michal Sedl\'ak for very useful suggestions, Helena Fikerov\'{a} for
her help with software for active stabilization of the interferometer, and Ivo
Straka for the construction of the two-photon source. This work was supported
by the Italian Ministry of Education through PRIN 2008 and the European
Community through the COQUIT project. M. Dall'Arno was founded by the Spanish project
FIS2010-14830. Experimental part was supported by the Czech Science Foundation
(202/09/0747), Palacky University (PrF-2011-015), and the Czech Ministry of
Education (MSM6198959213, LC06007).


\begin{thebibliography}{}
\bibitem{note:illumination} This fact
  is no longer true for example in the analogous context of quantum
  illumination \cite{Llo08, TEGGLMPS08, SL09}, where one's task is to perform
  a low energy detection of the presence (or absence) of a far object in a
  noisy environment.
\end{thebibliography}
\end{document}